\documentclass[twocolumn,preprintnumbers,amsmath,amssymb]{revtex4}

\usepackage{amsthm}
\usepackage{graphicx}
\usepackage{algorithm}
\usepackage{listings}
\usepackage{algorithmic}
\usepackage{xcolor}
\usepackage[utf8]{inputenc}
\usepackage{amsmath}
\usepackage{mathtools}

\newtheorem{theorem}{Theorem}
\newcommand{\commentoutA}[1]{}

\begin{document}
    
    \preprint{LA-UR-21-20350}
    
    \title{Mixed Precision Fermi-Operator Expansion on Tensor Cores From a Machine Learning Perspective}
    
    \author{Joshua Finkelstein$^{\dagger *}$, Justin Smith$^\dagger$, Susan M. Mniszewski$^\ddagger$, Kipton Barros$^\dagger$, Christian F. A. Negre$^{\dagger *}$, Emanuel H. Rubensson$^{\dagger \dagger}$, Anders M. N. Niklasson$^\dagger$}
    \email{jdf@lanl.gov, cnegre@lanl.gov, amn@lanl.gov}
    \affiliation{$^\dagger$Theoretical Division, Los Alamos National Laboratory, Los Alamos, New Mexico 87545}
    \affiliation{$^\ddagger$Computer, Computational, and Statistical Sciences Division, Los Alamos National Laboratory, Los Alamos, New Mexico 87545}
    \affiliation{$^{\dagger \dagger}$Division of Scientific Computing, Department of Information Technology, Uppsala University, Box 337, SE-751 05 Uppsala, Sweden}
    \date{\today}
    
    \begin{abstract}
        We present a second-order recursive Fermi-operator expansion scheme using mixed precision floating point operations to perform electronic structure calculations using tensor core units. A performance of over 100 teraFLOPs is achieved for half-precision floating point operations on Nvidia's A100 tensor core units. The second-order recursive Fermi-operator scheme is formulated in terms of a generalized, differentiable deep neural network structure, which solves the quantum mechanical electronic structure problem. We demonstrate how this network can be accelerated by optimizing the weight and bias values to substantially reduce the number of layers required for convergence. We also show how this machine learning approach can be used to optimize the  coefficients of the recursive Fermi-operator expansion to accurately represent fractional occupation numbers of the electronic states at finite temperatures.
    \end{abstract}
    
    \keywords{electronic structure theory, density functional theory, constrained density functional theory, Fermi-operator expansion, deep neural networks, machine learning, tensor core, mixed precision}
    \maketitle
    
    \section{Introduction}
    Electronic structure calculations based on Hartree-Fock, density-functional theory, or semiempirical methods often require the intermediate construction of the single-particle density matrix \cite{RMcWeeny56,RMcWeeny60,Finnis2003-qx}. This density-matrix can be calculated with different techniques, for example, the use of a direct diagonalization of the Kohn-Sham Hamiltonian or the Fockian \cite{ASzabo89}, Green's function methods \cite{Economou2006-od}, variational optimization \cite{Payne1982-yv}, and various recursive Fermi-operator expansion schemes \cite{RMcWeeny60,APalser98,Bowler1999-yu}. The method of choice often depends on several criteria such as: 1) the electronic structure basis set, i.e.,\ if plane waves or localized atomic orbitals are used; 2) the system that is analyzed, i.e.,\ if the system is small or large or if it is metallic or non-metallic; and 3) the computational platform, i.e.,\ if the calculation is performed on a single or multiple central processing units (CPUs) or on a hybrid architecture with graphics processing units (GPUs).
    
    In this article we target density matrix calculations for electronic structure methods on tensor core units with an adapted atomic-orbital-like basis set for intermediate sized nonmetallic systems. For these calculations, we use a second-order recursive Fermi-operator expansion scheme in combination with mixed precision floating point operations, which enables efficient calculations of the density matrix using tensor core accelerators \cite{nvda-tc}. Exploring the use of tensor core based architectures for electronic structure calculations follows the previous transitioning from CPU-only based electronic structure techniques to the more specialized GPU-based techniques \cite{tmartinez08,tmartinez09a,tmartinez09b,jstone10,tmartinez11,jmaia12,mhacene12,fliu2015,huhn20,gordon20,ZGuoqing20}. Our Fermi-operator expansion scheme is formulated and presented in terms of a generalized convolutional deep neural network \cite{jschmidhuber2015,higham2019}. This network formulation provides a powerful machine learning perspective on how we can further optimize and extend the applications of the recursive Fermi-operator expansion. We find that we can optimize the weight and bias values, and use a combination of multiple layers to represent Fermi functions at finite electronic temperatures with high numerical accuracy. We also find that an optimized set of weight and bias values can reduce the number of layers required to reach convergence. 
    
    The article is outlined as follows. First, we discuss the electronic structure problem and the Fermi-operator representation of the density matrix. Then, we present a second-order recursive Fermi-operator expansion method in terms of a generalized deep neural network using mixed precision floating-point operations that are well adapted for tensor core calculations. We then demonstrate and analyze the performance on tensor core units for some test examples. Thereafter, we discuss how optimized weight and bias values can be used to accelerate convergence and how they accurately represent the Fermi function for fractional occupation numbers at finite electronic temperatures. The algorithms are presented in pseudo-code throughout the manuscript and implementations in Python are available in the Supporting Information (SI) document.
    
    \section{Density-Matrix Fermi-operator Expansion}
    
    \subsection{The Density Matrix}
    \newpage
    
    The single-particle density matrix, $D$, is given in terms of the Fermi matrix function, where
    
    \begin{equation}
        D = \left(e^{\beta(H-\mu I)}+I\right)^{-1}.
    \end{equation}
    Here, $I$ is the identity matrix, $\beta= 1/(k_B T_e)$ is the inverse electronic temperature, $\mu$ is the chemical potential, and $H \in \mathbf{R}^{N\times N}$ is the Hamiltonian (or Fockian) matrix with matrix elements
    \begin{equation}
        H_{i,j} = \langle \phi_i \vert {\widehat H}\vert \phi_j\rangle.
    \end{equation}
    
    For simplicity, we assume an orthonormal basis-set $\{\phi_i\}_{i = 1}^N$, so that the overlap matrix $S_{i,j} = \langle \phi_i \vert \phi_j\rangle = \delta_{i,j}$. An orthonormalized basis-set representation can always be constructed from a congruence transform based on the inverse factorization of the overlap matrix \cite{cnegre16}. The operator ${\widehat H}$ is the effective single-particle Hamiltonian operator involved in methods such as Hartree-Fock or Kohn-Sham density functional theory. At zero electronic temperature, $T_e = 0$, the Fermi function becomes a Heaviside step function and the density matrix reads as:
    \begin{equation}
        D = \theta(\mu I - H).
    \end{equation}
    There are several methods that can be used to calculate the density matrix. The traditional method is based on a direct diagonalization of $H$, i.e.\ finding the orthonormal eigenstates $Q$ such that
    \begin{equation}
        Q^T H Q = E, ~~ E_{ij} = \epsilon_i \delta_{ij}, ~~ Q^TQ = I.
    \end{equation}
   In the diagonal (eigenvector) representation, $E$, and the identity matrix, $I$, are diagonal and the matrix exponential and inversion can therefore be calculated directly, i.e.
    \begin{equation}\begin{array}{ll}
            D & = QQ^T \left(e^{\beta(H-\mu I)}+I\right)^{-1}QQ^T \\
            ~~\\
            &  = Q\left(e^{\beta(E-\mu I)}+I\right)^{-1}Q^T,
        \end{array}
    \end{equation}
    where the chemical potential $\mu$ is adjusted to account for the desired orbital occupation, $N_{\rm occ}$, i.e.\ such that ${\rm Tr}[D] = N_{\rm occ}$.
    Alternatively, in the zero-temperature limit the density matrix becomes
    \begin{equation}\begin{array}{ll}
            D & = QQ^T \theta(\mu I - H) QQ^T \\
            ~~\\
            &  = Q\theta(\mu I - E)Q^T,
        \end{array}
    \end{equation}
    where the shifted Heaviside step function, $\theta(\mu I - E)$, is evaluated on each of the diagonal elements of $E$ (eigenvalues of $H$). The chemical potential needs to be shifted to reach a desired occupation also in this case.
    
    \subsection{Serial Fermi-Operator Expansions}
    
    An alternative to construct the density matrix, $D$, is the serial Chebyshev Fermi-operator expansion scheme \cite{rsilver94,AWeisse06}, where the density matrix is approximated by a linear combination of Chebyshev matrix polynomials of the Hamiltonian, $T_n(H)$, 
    \begin{equation}
        D = \left(e^{\beta(H-\mu I)}+I\right)^{-1} \approx \sum_{n=1}^m c_n T_n(H).
    \end{equation}
    Alternatively, we may also use a Green's function expansion, which is based on a complex contour integration \cite{RZeller85,nbernstein01,sgoedecker93,tozaki07,LLin13} with some complex energy mesh $\{z_n\}$, where
    \begin{equation}
        D = \left(e^{\beta(H-\mu I)}+I\right)^{-1} \approx \sum_{n=1}^m c_n \left(H - z_n I\right)^{-1}.
    \end{equation}
    
    In both of these serial Fermi-operator expansion methods, the coefficients $\{c_n\}_{n=1}^m$, need to be adjusted such that the approximate density matrix has the correct occupation and temperature. To reach accurate representations, high-order expansions are required and convergence can be hard to achieve at low temperatures. Higher-order polynomials are needed for the Chebyshev expansion and the Green's function expansion requires complex energies close to the real axis which may lead to singularity problems for low-temperature expansions. However, the Chebyshev and Green's function methods can take advantage of sparse matrix algebra for sufficiently large Hamiltonian matrices, which allows computations with linear scaling complexity \cite{SGoedecker99,DBowler12}. Chebyshev methods can sometimes also take advantage of a stochastic sampling of expectation values using a smaller randomized trial basis. In these cases the ${\cal O}(N^3)$ cubic scaling cost of the diagonalization, a function of the system size $N$, can be replaced by a linear ${\cal O}(N)$ scaling complexity \cite{RSilver96,SGoedecker99,DBowler12,LLin13}. This is a particular advantage for very large problems such as those including tens of thousands of atoms or electrons. For smaller problems, the construction of the density matrix with direct diagonalization is typically much faster. 
    
    \subsection{Recursive Fermi-Operator Expansion}
    
    At zero electronic temperature a Fermi-operator expansion scheme has to approximate the Heaviside step function of $H$, with the step formed at the chemical potential. In this zero-temperature limit, Chebyshev and Green's function methods can have convergence problems \cite{rsilver94,AWeisse06}. An alternative is given by a recursive Fermi-operator expansion \cite{APalser98,KNemeth00,AHolas01,ANiklasson02,ANiklasson03b,DKJordan05,ERudberg11,PSuryanarayana13,EHRubensson14,Truflandier16}, where
    \begin{equation} \label{Recursion}
        D = \theta(\mu I - H) \approx f_m(f_{m-1}(\ldots f_0(H)\ldots )).
    \end{equation}
    The recursion can be performed by successive projections of a matrix $X_i$ that starts with $X_0 = H$ and then $X_{i+1} = f_i(X_i)$ is calculated in each iteration until convergence is reached as $X_i \rightarrow D$. The functions $f_i(X_i)$ are chosen to {\em project} the eigenvalue spectrum of $X_i$ to a more {\em pure} ensemble with eigenvalues closer either to 1 or to 0. The occupation is 1 for the occupied states below the chemical potential $\mu$ and 0 for the unoccupied states above. These type of recursive Fermi-operator expansion methods are also referred to as {\em purification} or {\em spectral projection} schemes \cite{RMcWeeny60,GBeylkin99,ANiklasson02}. 
    
    The advantage of a recursive Fermi-operator expansion is that we can reach a very high polynomial order in the approximation with only a few number of iterations. There are many choices of polynomials and techniques to adjust the expansion such that the step is formed at the chemical potential. Possibly the simplest and most efficient technique is the second-order spectral projection (SP2) method \cite{ANiklasson02,EHRubensson11,ERudberg11,EHRubensson14}, which is the main focus of this article.
    
    \section{The SP2 method}
    
    \subsection{Second-Order Spectral Projection Polynomials}
    
    In the SP2 method, the recursive expansion functions in Eq.\ (\ref{Recursion}) are chosen as second-order polynomials acting on the interval $[0,1]$. In the original version of SP2, 
    \begin{equation}\label{SP2_Pol}
        X_{i+1} = f_i(X_i) = X_i \pm (X_i - X_i^2).
    \end{equation}
    The $\pm$ sign is chosen to adjust the trace of $X_{i+1}$ in each projection such that the correct occupation, $N_{\rm occ}$, is reached at convergence, i.e.\ such that  ${\rm Tr}[X_i] \rightarrow {\rm Tr}[D] = N_{\rm occ}$ \cite{ANiklasson02}. In this way the step is formed automatically at the correct chemical potential $\mu$. No prior knowledge of the chemical potential is therefore required and no post-processing adjustment is needed. The polynomial expansion order doubles in each recursion. In only 30 recursion steps the polynomial expansion order is over one billion. The second-order polynomials in Eq.\ (\ref{SP2_Pol}) are continuously increasing and decreasing functions on the expansion interval $[0,1]$. The expansion therefore automatically avoids any type of Gibbs oscillations that are sometimes a problem in Chebyshev expansion methods \cite{SILVER1996115,RevModPhys.78.275}. A truncated version of the SP2 scheme, where the recursion is terminated before the convergence to an idempotent solution with integer occupation numbers is reached, can also be used to approximate the Fermi-operator at elevated electronic temperatures \cite{smmniszewski19}. 
    
    The second-order polynomial projection functions in Eq.\ (\ref{SP2_Pol}) can be modified by a shift-and-scale transformation that accelerates the expansion \cite{EHRubensson11}. To guarantee stability and convergence, this acceleration technique requires prior knowledge of the two eigenvalues right above and right below the chemical potential (the eigenvalues corresponding to the highest occupied (HOMO) and lowest unoccupied (LUMO) molecular orbitals) or at least some fairly accurate estimate of their values \cite{EHRubensson14}. For repeated applications of the SP2 algorithm, which is necessary, e.g.\ in molecular dynamics simulations, rigorous estimates of the HOMO and LUMO eigenvalues can be calculated from each previous SP2 Fermi-operator expansion \cite{EHRubensson14}. In this case only the first SP2 expansion can't use the shift-and-scale acceleration technique because there is no prior knowledge of the HOMO and LUMO energies. 
    
    In our generalized deep-neural network representation of the SP2 scheme, presented below, we will show how an accelerated convergence can be obtained from an optimization of the weight and bias values of the network. This machine learning perspective can also be used to optimize the coefficients and re-weight the different layers of the SP2 expansion to get a highly accurate representation of the Fermi function at finite electronic temperatures with fractional occupation numbers. This offers a significant improvement over the truncated SP2 scheme that has previously been used to approximate the Fermi function at finite temperatures \cite{smmniszewski19}.
    
    \subsection{Deep-NN SP2}
    
    There are several ways to implement the SP2 recursive Fermi-operator expansion of Eq.\ (\ref{Recursion}), using the second-order projection polynomials in Eq.\ (\ref{SP2_Pol}). Here we will present a version that naturally maps onto the algorithmic structure of a generalized convolutional deep neural network (Deep-NN). 
    
    The original SP2 expansion has two key properties if we assume all eigenvalues of $X_1 \in [0,1]$: 1) it converges to a step function with the step formed somewhere in the interval $[0,1]$; and 2) each projection step either increases or decreases the trace of $X_i$ by projecting the eigenstates either toward a stationary point at $1$ or toward a stationary point at $0$. An initial linear transform, $X_1 = f_0(H)$, is chosen to scale the eigenstates of $H$ to the interval $[0,1]$ in reverse order. Following this initial transform, we can choose the projection polynomials that improve the convergence of the trace after each recursion. When the trace corrections no longer improve the occupation or when all the eigenvalues of $X_i$ are as close as possible to $0$ or $1$, convergence has been reached and the expansion is terminated. At this point we may use the converged density matrix $X_m$ to calculate various quantum mechanical observables, $\langle A \rangle = \mathrm{Tr}[X_mA]$, where $A$ is the matrix representation of the relevant operator, e.g., the Hamiltonian matrix, $H$, for the energy. It is easy to see how this scheme can be reformulated and mapped onto the structure of a Deep-NN as shown in Fig.\ 3. In the first layer we use the Hamiltonian $X_0 = H$ as the input descriptor. The weight and bias functions, $W_0$ and $B_0$, are then chosen such that $S_0$ is the rescaled Hamiltonian with eigenvalues in reverse order inside the interval $[0,1]$. As an activation function we chose the {\em matrix function} $f(S_0) = S_0^2$, which acts on the eigenvalues of $S_0$. This is in contrast to regular neural networks where the activation function acts on the individual matrix elements. The matrix square operation of the activation function consists of a single tensor contraction, i.e.\ a matrix-matrix multiplication, which is an advantage since tensor cores are optimized to perform such operations at high speed. At the subsequent layer, where $X_1 = f(S_0)$, we chose the weight and bias values such that $S_1 = W_1X_1 + B_1$, with $W_1 = \sigma_1 I$ and $B_1 = (I-W_1)S_0$. The value of $\sigma_1 = \pm 1$ is chosen such that $S_1$ has the smallest occupation error, $|{\rm Tr}[S_1] - N_{\rm occ}|$, of the two sign alternatives. These operations are continued layer by layer and the $i$-th approximation to the density matrix is computed as follows:
    \begin{equation}
        X_i = f(\ldots f(W_1f(W_0X_0+B_0) + B_1) \ldots)~.
    \end{equation}
    At the last layer, $S_{m-1}$, once the occupation error has converged to some sufficiently accurate value, the density matrix is outputted as, $D = X_m$. 
    
    The Deep-NN formulation of the SP2 algorithm is given in pseudocode in Alg.\ \ref{Deep_NN_SP2} and also includes a {\em parameter-free} check for convergence. 
    The convergence is determined from where an expected decrease, under exact arithmetics, of the {\em estimated} idempotency error, ${\rm IdErr}_n$, is no longer fulfilled in practice.
    A motivation for and precise derivation of the convergence criterion is provided in the appendix. Typically, the idempotency error is computed as $\|X-X^2\|$, where $\|\cdot\|$ is either the spectral (2-norm) or the Frobenius norm. Here we have instead chosen to use $\textrm{Tr}[X-X^2]$ as the measure of the idempotency error, which is a simpler and more computationally efficient measure. In fact, since $\|X-X^2\|_2 \le \textrm{Tr}[X-X^2] \le N \|X-X^2\|_2$ whenever the eigenvalues of $X$ are in $[0,1]$, convergence in $\textrm{Tr}[X-X^2]$ is equivalent to convergence in the spectral norm. Only a single $\mathcal{O}(N)$ trace operation is needed in each deep layer. 
    
    In Alg.\ \ref{Deep_NN_SP2} we also include a small constant $\epsilon$ which ensures that we get an alternating sign of $\sigma_n$ if the occupation corrections are very small, inducing faster convergence. Otherwise, the sign, $\sigma_n$, is chosen to minimize the occupation error in $S_n$. The inclusion of the small $\epsilon$ term is an ad-hoc adjustment that, in general, is not a necessity of the convergence criteria. A Python script implementing the full Deep-NN SP2 algorithm is presented in the supplementary material (SI).
    
     \begin{figure}
        \includegraphics[scale=0.60]{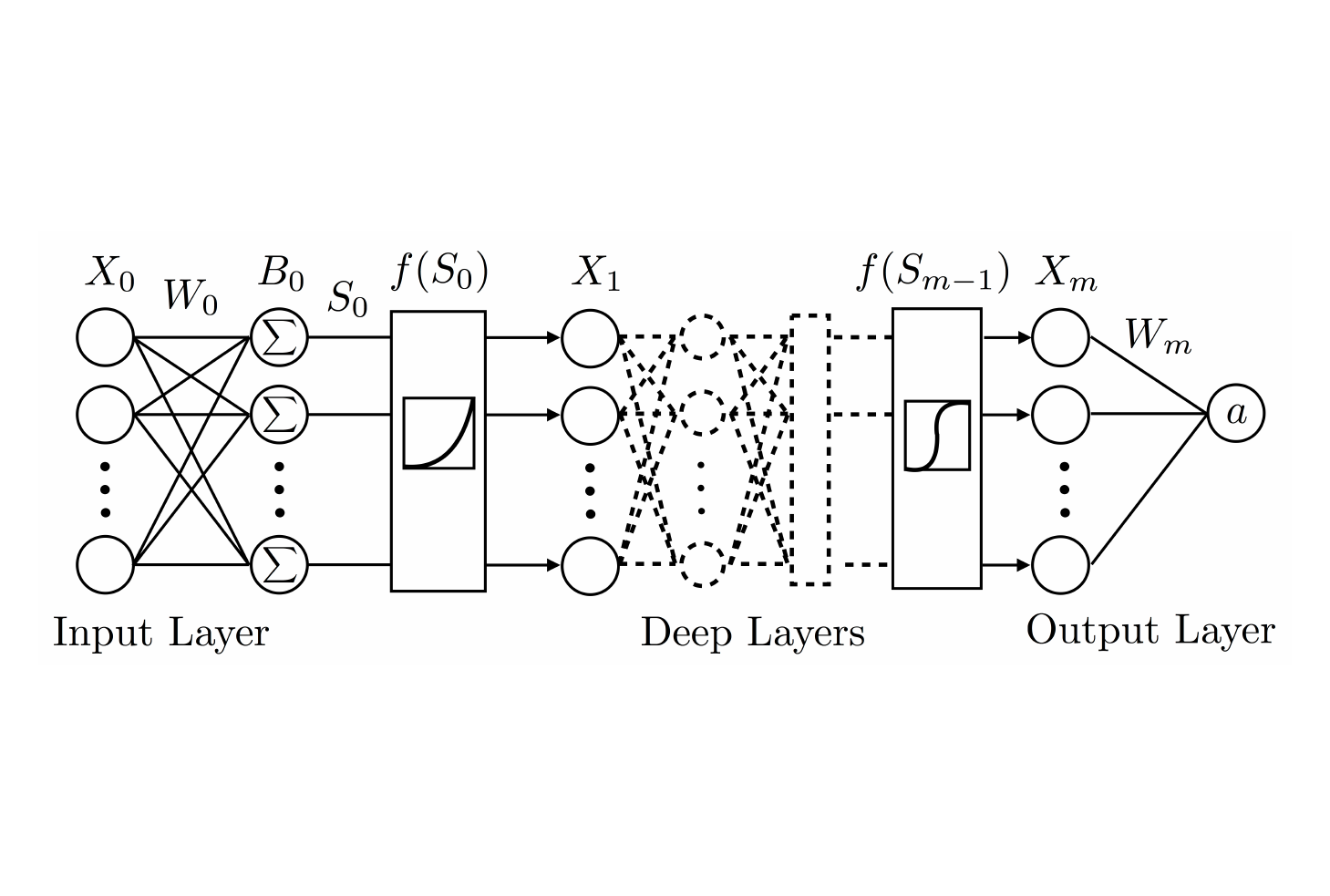}
        \caption{\label{Fig_NN} Schematic picture of a deep neural network.
            The weights $W_n$ and bias values $B_n$ generate a linear transformation of $X_n$, where $S_n = W_nX_n + B_n$,  are given in standard matrix notation. A new layer $X_{n+1}$ is provided after the application of a non-linear activation function, i.e.\ where $X_{n+1} = f(S_n)$. }
    \end{figure}
    
    \begin{algorithm}
    \caption{{\small The Deep-NN formulation of the SP2 recursive Fermi-operator expansion algorithm. A Python script is provided in the supplementary information.}}
    \label{Deep_NN_SP2}
    \algsetup{indent=1em}
    \begin{algorithmic}
        \STATE $ N_{\rm occ}, ~ \mbox{Number of occupied states or orbitals}$
        \STATE $\epsilon, ~\mbox{small number close (or equal) to 0}$
        \STATE $ H, ~  \mbox{Orthonormalized Hamiltonian}$
        \STATE $ \varepsilon_1, ~\varepsilon_N, ~\mbox{Spectral bound estimates of } H$
        \STATE $ W_n, ~  \mbox{Observable operator matrix, e.g.}~ W_n = H$ 
        \STATE $ X_0 = H, ~\mbox{Input layer}$
        \STATE $ W_0 = - (\varepsilon_N - \varepsilon_1)^{-1}I, ~~B_0 = \varepsilon_N(\varepsilon_N-\varepsilon_1)^{-1}I$
        \STATE $S_0 = W_0X_0 + B_0, ~~N_S = {\rm Tr}[S_0], ~\mbox{Occupation of}~ S_0$
        \STATE $ n = 0, ~~\mbox{Number of layers}$
        \WHILE{Not Converged}
            \STATE $ n = n + 1$
            \STATE $ X_n = f(S_{n-1})$, ~~\mbox{See Alg.\ \ref{ActivationFunction}} 
            \STATE $N_X = {\rm Tr}[X_n]$
            \STATE ${\rm IdErr}_n = N_S - N_X, ~~\mbox{Idempotency error estimate}$
            \STATE $\sigma_n = 
            \mbox{Sign} 
            \left( |2N_S - N_X - 
            N_{\rm occ}| 
            - | N_X -N_{\rm occ}| -\sigma_{n-1} \epsilon \right) $
            \STATE $ W_n = \sigma_n I, ~ B_n = (I-W_n)S_{n-1}$ 
            \STATE $ S_n = W_nX_n + B_n $
            \STATE $N_S= W_nN_X + (1-\sigma_n)N_S,~ \mbox{Updated occupation of}~ S_n$ 
            \IF{${\rm IdErr}_n <= 0$}
              \STATE {\rm Converged} = \TRUE
              \ELSIF{$n > 2$ \AND $\sigma_{n-1} \ne \sigma_{n-2}$ \AND ${\rm IdErr}_{n} > 4.5 \times ({\rm IdErr}_{n-2})^2$}
              \STATE {\rm Converged} = \TRUE
            \ENDIF
        \ENDWHILE
        \STATE $D = f_{\rm final}(S_{m-1}), ~\mbox{final}~ X_m~ \mbox{layer in Fig.\  \ref{Fig_NN}}$
        \STATE $\alpha  = {\rm Tr}[DW_m], ~ \mbox{Output observable}$
        \end{algorithmic}
    \end{algorithm}
    
    \subsection{Mixed Precision Operations}
    
    The computationally dominant step in the Deep-NN SP2 Fermi-operator expansion is the calculation of the matrix square in the activation function. In dense matrix algebra, such generalized matrix-matrix multiplications can often be performed with very high performance on almost any computational platform. The SP2 scheme therefore stands out as an efficient alternative to diagonalization-based density matrix calculations. Here we are interested in using tensor core calculations. Tensor core units have been tailored to perform tensor contractions, i.e.\ matrix-matrix multiplications for machine learning applications using convolutional deep neural networks with close to peak performance. Recently, Nvidia's V100 tensor core accelerated graphics processing unit broke the 100 teraFLOPs barrier for deep learning applications \cite{nvda-v100}. Our goal is to use such tensor core accelerators for the calculation of density matrices using the Deep-NN SP2 Fermi-operator expansion. 
    
    The tensor core units use low, mixed precision floating point operations. Typically, only half-precision operations with single-precision accumulation are used. The half-precision is in general too low in accuracy for meaningful density matrix calculations, but a single precision accuracy is good enough for many problems.  To achieve single precision accuracy we can represent a single-precision matrix $X$ with a pair of two half-precision matrices,
    \begin{equation}\label{Split1}
        X \approx X^{(0)} + X^{(1)}.
    \end{equation}
    Using pseudocode notation, the corresponding dual half-precision representation of a matrix $X$ would be generated by
    \begin{equation}\label{Split2}
        \begin{array}{l}
            X^{(0)}  = {\rm FP16}[X] \\
            ~~\\
            X^{(1)}  = {\rm FP16}[X - X^{(0)}],
        \end{array}
    \end{equation}
    where ${\rm FP16}[~]$ denotes the half-precision representation.
    Generalizations to a higher level of accuracy using multiple matrices, $X^{(n)}$, is straightforward and will not be discussed. A matrix-matrix multiplication, $X \times Y$, can then be performed using four separate matrix-matrix multiplications in half precision with accumulation in single precision (${\rm FP32}[~]$), i.e.
    \begin{equation} \label{XY} \begin{array}{l}
            X \times Y \approx 
            {\rm FP32}\left[\left(X^{(0)} + X^{(1)}\right)\left(Y^{(0)} + Y^{(1)}\right)\right]\\
            ~~\\
            = {\rm FP32}\left[(X^{(0)} Y^{(0)} + \left(X^{(1)} Y^{(0)} + X^{(0)} Y^{(1)} \right)\right.\\
            ~~\\
            ~~~~~~~~~~~~ \left. + X^{(1)} Y^{(1)}\right].
        \end{array}
    \end{equation}
    In the Deep-NN SP2 Fermi-operator scheme in Alg.\ \ref{Deep_NN_SP2} we only need to calculate matrix squares in the activation function. If we assume that each matrix $X$ is symmetric and neglect the small $X^{(1)}Y^{(1)}$-term we can reduce the calculation of a matrix square to only two matrix-matrix multiplications in half-precision and single accumulation, i.e.\
    \begin{equation}\label{Square}\begin{array}{l}
            X^2 \approx {\rm FP32}\left[X^{(0)} X^{(0)} + X^{(0)} X^{(1)} + (X^{(0)} X^{(1)})^T \right].\\
        \end{array}
    \end{equation}
    All the matrix products and sums are assumed to be accumulated in single precision (FP32). This approach would also benefit from multiplications of symmetric matrices where only the upper or lower half matrix needs to be calculated.
    
    \subsection{Mixed Precision Deep-NN SP2}
    
    To adjust the Deep-NN SP2 algorithm in Alg.\ \ref{Deep_NN_SP2} to mixed precision floating-point operations, we only need to adjust the activation function, $f(X) = X^2$, where the matrix square is performed using tensor contractions on tensor core units in half precision with single accumulation. This is described by
    Alg.\ \ref{ActivationFunction}.
    
    \begin{algorithm}
    \caption{{\small Calculation of the activation function, $f(X) = X^2$, in the Deep-NN SP2 scheme in Alg.\ \ref{Deep_NN_SP2} for tensor core units in half precision and single accumulation.}}
    \label{ActivationFunction}
    \algsetup{indent=1em}
    \begin{algorithmic}
        \STATE $ X ~~ \mbox{Input matrix in single precision}$
        \STATE $ X^{(0)} = {\rm FP16}[X]$
        \STATE $ X^{(1)} = {\rm FP16}[X-X^{(0)}]$
        \STATE $A = {\rm FP32}[X^{(0)} \times X^{(0)}] ~~\mbox{Tensor core  multiplication}$
        \STATE $B = {\rm FP32}[X^{(0)} \times X^{(1)}] ~~\mbox{Tensor core multiplication} $
        \STATE $f(X) = X^2 \approx {\rm FP32}[A + B + B^T]$
    \end{algorithmic}
    \end{algorithm}
    
    The main source of the error in the mixed precision Deep-NN SP2 scheme is the eigenvalue distribution at convergence. Because of the finite precision, the eigenstates will not be exactly 1 or 0, corresponding to fully occupied or unoccupied states. This may lead to significant errors in energy calculations. However, these errors can be reduced by a post-processing step. The post-processing refinement can be achieved by using a modified activation function in the final step, $f_{\rm final}(S_n)$, in Alg.\ \ref{Deep_NN_SP2}. Instead of the matrix square, we use 
    \begin{equation}\label{Refinement}
        f_{\rm final}(S_n) = \begin{cases}
        2S_n^2 - S_n^4, &\text{ if } \sigma_{n-2} = 1 \\
        (2S_n - S_n^2)^2, &\text{ if } \sigma_{n-2} = -1
        \end{cases}
    \end{equation}
    which is calculated in an enhanced precision, either standard double precision or in single precision with double accumulation.  
    
    Figure \ref{ConvergenceDeepNN} shows the convergence in the energy, ${\rm Tr}[X_nH]$, with the error compared to the ``exact'' energy, and the idempotency error measured by the spectral norm, $\|X_n^2-X_n\|_2$. The rapid improvement by multiple orders of magnitude in the last layer is given by the final refinement step performed in an enhanced precision which scales quadratically from $n-1$ to $n$. The enhanced precision is often necessary to attain a sufficiently high numerical accuracy.
        
         \begin{figure}
        \includegraphics[scale=0.31]{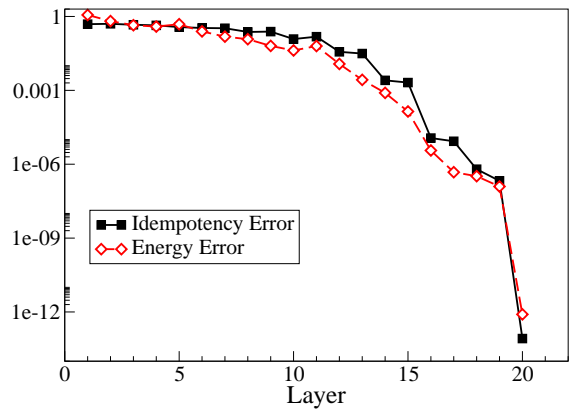}
        \caption{\label{ConvergenceDeepNN}
            {\small The idempotency convergence of the Deep-NN SP2 Fermi operator expansion scheme for a small uniformly randomized symmetric Hamiltonian ($H \in \mathbf{R}^{10 \times 10}, N_{\rm occ} = 5,  H_{ij} \in [-1,1]$) using the mixed precision algorithm. The last layer corresponds to the final refinement activation function.  }}
    \end{figure}

    \subsection{Convergence Estimate For Low Precision Floating-Point Operations}\label{ConvC}
    
    Using only half-precision floating point operations in the Deep-NN SP2 scheme leads to fairly large errors compared to regular double precision operations, even if the dual mixed precision, presented above, is used. Thanks to the post-processing refinement step the final error is significantly reduced. However, we first need to determine when convergence is reached. This can be difficult to decide under numerically noisy conditions caused by the low-precision floating point operations. The idea we use to determine convergence is based on the observation that the idempotency estimate we use in Alg.\ \ref{Deep_NN_SP2}, i.e.\ ${\rm IdErr}_n = {\rm Tr}[S_{n-1}-X_n] = {\rm Tr}[S_{n-1}-S_{n-1}^2]  $,  decreases quadratically between every second step if we have alternating signs of $\sigma_n$. However, limitations in the finite precision will, at some point, prevent the expected decay of the idempotency error. At this point, the expansion can then be terminated, because the best possible convergence has been reached. This parameter-free convergence estimate is both efficient and easy to implement. 
    
    The refinement in Eq.\ (\ref{Refinement}) is the result of composing two layers in a single step with $\sigma_n = 1$ and $\sigma_{n-1}=-1$ or $\sigma_n = -1$ and $\sigma_{n-1}=1$. These alternating  signs of $\sigma$ provide for a guaranteed second order decrease in the error if exact floating point operations are used. Our convergence criterion is analogous to the parameterless stopping criteria by Kruchinina, Rudberg and Rubensson \cite{AKruchinina16}, which here has been adapted to a different idempotency measure. Details of the derivation are given in the appendix. 

    \section{Mixed Precision Fermi-operator expansion on tensor cores}

    To make use of tensor cores with matrix multiplications, the Deep-NN SP2 algorithm was written using CUDA v11.0, the cuBLAS library and several customized kernels. All matrix multiplications in the SP2 algorithm were carried out using cuBLAS general matrix multiplication (GEMM) calls. The GEMM calls execute tensor core operations automatically and no special commands are required to make use of them. This automatic feature can be disabled with the appropriate cuBLAS API call. 
    
    Implementation of the half precision multiplications needed by the $X^2$ activation function in Alg.\ \ref{ActivationFunction} required several custom kernels. These kernels decompose the matrix $X$ into a sum of two FP16 matrices that are then multiplied and summed as described in Eqs.\ (\ref{Split1}) through (\ref{Square}) using standard cuBLAS GEMM routines. Custom kernels were also necessary to reduce, as much as possible, the amount of data transfer between the host and GPU device memory. The Deep-NN SP2 CUDA implementation will be made available through the PROGRESS \cite{2016progress} library.
    
    The rate of floating point operations (FLOP) for the Deep-NN SP2 algorithm was estimated from simulations using tensor cores on both Nvidia A100 and V100 GPUs and is shown in Fig.\ \ref{flops_vs_N}. This estimate does not include the initialization and memory allocation of the routine nor the final layer, i.e. the double precision refinement step. For purposes of comparison, this FLOP rate was also computed on the V100 with the tensor cores \emph{disabled}, we call this the GPU-only FLOP rate; it is displayed in Fig.\ \ref{flops_vs_N} as well. We observe an approximate 7-8x speed up on the V100 when tensor cores are enabled versus when they are disabled and only the GPU is used. Even more impressive, we achieve approximately 120 teraFLOPs on the A100 when utilizing tensor cores.

    Although the plots in Fig.\ \ref{flops_vs_N} suggest the Deep-NN SP2 algorithm may only be beneficial for large $N$ values, a recent publication \cite{Abdelfattah2019-pb} shows how matrix-matrix multiplications can reach high performance also for smaller $N$ by utilizing batching techniques. The same technique would most likely benefit the SP2 method for small $N$. Additionally, further performance increases might also be gained by considering a sparse matrix implementation \cite{ozachariadis20} of the Deep-NN SP2 method. 
     
    \begin{figure}
        \includegraphics[scale=0.31]{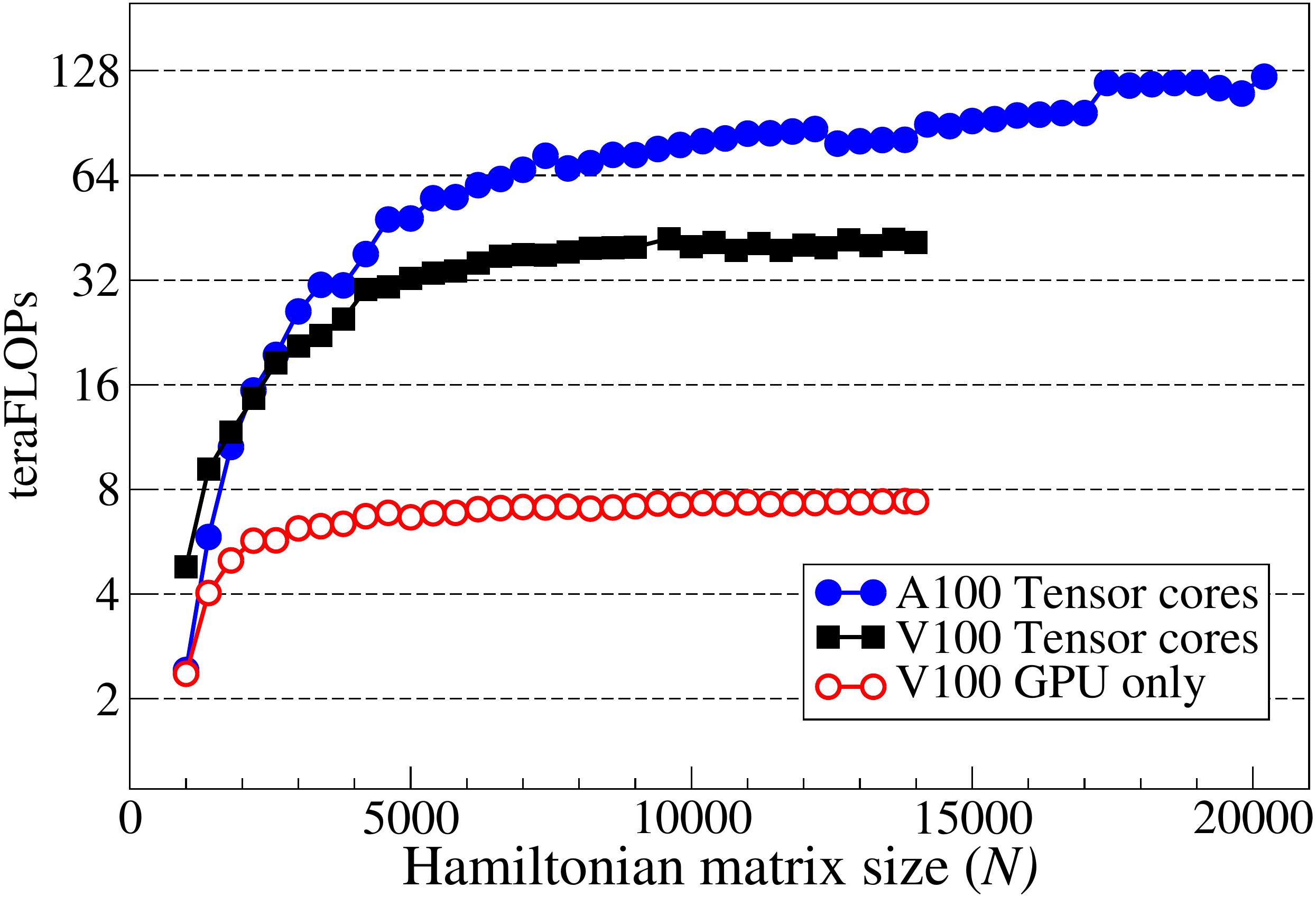}
        \caption{\label{flops_vs_N}
            {\small Half-precision tera-floating point operations per second (teraFLOPs) vs.\ Hamiltonian system size for the Deep-NN SP2 Fermi-operator expansion running on Nvidia's Volta V100 tensor core units, on its Volta GPU only, and on the more recent A100 tensor core units. The difference in maximum $N$ values is due to device memory limitations. }}
    \end{figure}
  
    \section{Capitalizing on the machine learning perspective}

    There are several observations that appear from the machine learning perspective of the SP2 Fermi-operator expansion scheme when it is formulated in terms of a layered network structure: 1) The quantum mechanical problem is solved naturally and with high efficiency through the computational structure of a generalized deep neural network; 2) The bias and weight values could be optimized using machine learning techniques to achieve improved convergence and possibly higher accuracy; 3) Other functions besides the matrix Heaviside step function could potentially be approximated through the same generalized deep neural network, including Fermi functions at finite electronic temperatures; 4) Recursive Fermi-operator schemes or sign-matrix expansions based on higher-order spectral projection polynomials could be mapped onto the same generalized network structure and use the same mixed precision technique; and 5) A recursive calculation of Green's Functions via a Dyson series expansion could also be generalized to fit into the algorithmic structure of Deep-NN SP2, e.g. $G = G_0 + G_0 V G$ (where $G$, $G_O$ and $V$ are the Green Function, the initial Green Function, and a perturbation to the Hamiltonian, respectively) can be rewritten recursively in a similar way to the SP2 scheme, where the corresponding weights and bias values could be optimized for convergence. 
    
    Here we will briefly discuss the ability to accelerate convergence for the Deep-NN SP2 algorithm and how approximate Fermi functions for fractional occupation numbers at elevated electronic temperatures can be generated recursively with high accuracy.
 
    \subsection{Accelerated Deep-NN SP2}
    
    In machine learning we try to learn the weight and bias functions by optimizing a regularized penalty function based on, for example, some large set of predetermined data. Here we may instead use the convergence rate to the idempotent density matrix. Each layer of the Deep-NN SP2 scheme can then be seen as generalized spectral projections with weights $W_n$ and bias values $B_n$. Instead of choosing the spectral projections from $W_n = \sigma_n I$, where $\sigma_n = \pm 1$, we may optimize over a continuous set of values, $\sigma_n \in \mathbf{R}$, as illustrated in Fig.\ \ref{AccDeepNNSP2}. To optimize convergence we chose the values of $W_n$, which in each separate layer gives the highest slope of the projection around the re-scaled eigenvalues corresponding to the HOMO or LUMO eigenvalues, but without risking switching places between occupied and unoccupied eigenvalues. This local choice of optimization accelerates the separation of the HOMO and LUMO eigenstates, which are the last to reach the fixed points at 1 and 0. This optimization requires prior knowledge of the re-scaled HOMO and LUMO eigenvalues. The optimized spectral projections may push eigenvalues outside of the $[0,1]$ interval, which could lead to divergence. To avoid this we need to shift and re-scale the eigenvalue spectrum to $[0,1]$ after each optimized projection. The combined transform from the choice of $\sigma_n$-values, followed by the shift and re-scaling, determines the optimized  weight, $W_n$, and bias values, $B_n$, in each layer. This local optimization of the weight and bias values of each layer can lead to a significant acceleration. Our accelerated Deep-NN SP2 algorithm is presented as a Python script in the supplementary information and the optimized choices of $W_n$ and $B_n$ for the network defined by $X_{n+1} = f(W_nX_n+B_n)$ are given there explicitly. This accelerated Deep-NN SP2 scheme turns out to be an equivalent Deep-NN formulation of the accelerated SP2 Fermi-operator expansion by Rubensson, which uses a shift and re-scale technique \cite{EHRubensson11,EHRubensson14}. However, here we arrive at the same acceleration scheme, but based on the Deep-NN perspective and with a different combination of spectral projection polynomials and choice of shift and re-scale transformations.
    
    An example of the convergence accelerated Deep-NN SP2 scheme is shown in Fig.\ \ref{AccConv}. The convergence is reached after 17 layers instead of 28, which is a significant improvement. 
    A disadvantage with the accelerated Deep-NN SP2 scheme is that it requires prior knowledge of the HOMO-LUMO eigenvalues. However, for repeated calculations of the density matrix, for example, in molecular dynamics simulations, the HOMO-LUMO eigenvalues can be estimated from previous time steps with a high-level of accuracy \cite{EHRubensson14}. Further acceleration of the Deep-NN SP2 scheme can possibly be achieved by tailoring the optimization of the weight and bias values for Hamiltonian matrices with particular eigenvalue distributions.
    
    \begin{figure}  
        \includegraphics[scale=0.31]{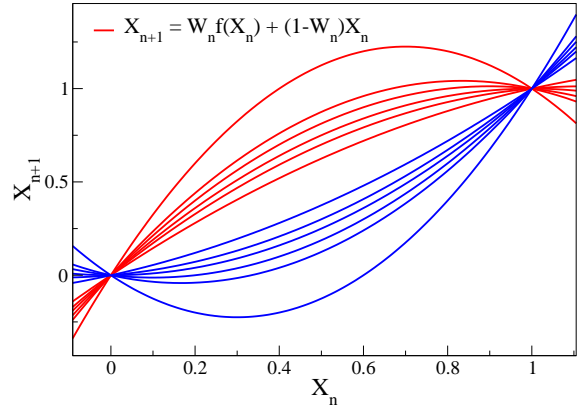}
        \caption{\label{AccDeepNNSP2}
            {\small Generalized Deep-NN SP2 projections in each layer, $X_n \rightarrow X_{n+1}$, with various weight values $W_n = \sigma_n I$, with $\sigma_n \in \mathbf{R}$, instead of $\sigma_n = \pm 1$. By locally optimizing $\sigma_n$ in each layer and using a shift and re-scale transforms to keep eigenvalues within the interval $[0,1]$, new weight and bias values for the {\em accelerated} Deep-NN formulation of the SP2 expansion are generated. A Python script for the accelerated Deep-NN SP2 algorithm is presented in the supplementary information.}}
    \end{figure}
    
    \begin{figure}
        \includegraphics[scale=0.31]{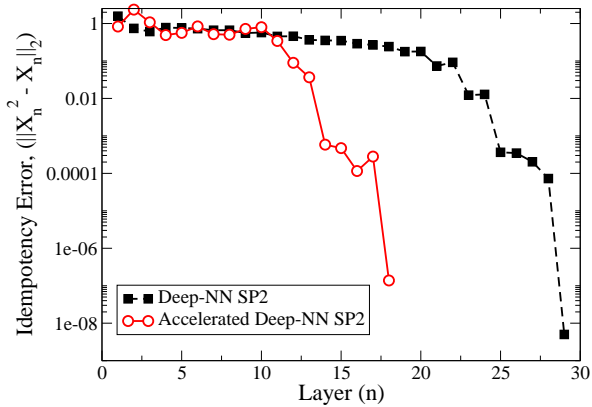}
        \caption{\label{AccConv}
            {\small The idempotency error, $\|X_n^2-X_n\|_2$, for a test Hamiltonian, $H \in \mathbf{R}^{100 \times 100}$, with $N_{\rm occ} = 10$, with and without acceleration.  A Python script of the accelerated Deep-NN SP2 scheme is presented in the supplementary material. }}
    \end{figure}

    \subsection{Optimized SP2 Finite Temperature Fermi-Operator Expansion}
    
    A recursive SP2 expansion that is stopped before it has reached convergence generates a smooth approximation to the Heaviside step function $\theta(\varepsilon - \mu)$ using the eigenvalues $\varepsilon_i \in [0,1]$ of the re-scaled Hamiltonian matrix $H$. This occurs because the second-order spectral projection functions are smooth on the interval $[0,1]$. These truncated SP2 expansions have previously been used to approximate the Fermi function at elevated electronic temperatures \cite{smmniszewski19}. However, seeing the truncated SP2 scheme in terms of a deep neural network allows for a straightforward optimization.  By generalizing the second-order spectral projection functions to a more general second-order polynomial and then optimizing the coefficients in each layer, we may achieve more accurate approximations of the Fermi function than those available to a truncated SP2 scheme alone. In this way, we can also optimize the convergence rate and minimize the error compared to an exact Fermi function. 

    Instead of using the alternating SP2 projection polynomials $x^2$ and $2x-x^2$, as is in  Eq.\ (\ref{SP2_Pol}), we allow for general second degree polynomials on $[0,1]$ to generate an approximation to the Fermi function at any $\varepsilon \in [0,1]$. We define the initial and $n$-th layer to be,
    \begin{align} 
    \begin{split} \label{general SP2}
        \varepsilon_0 &= \varepsilon \;,\\
        \varepsilon_{n} &= \theta_{n-1,1} \varepsilon_{n-1}^2 + \theta_{n-1,2} \varepsilon_{n-1} + \theta_{n-1,3}\;.
    \end{split}
    \end{align}
    To increase model flexibility, a linear combination of the intermediate values, $\sum_{i=0}^n c_i \varepsilon_i$, is used to enhance the approximation. Taking into account the folding of the eigenspectrum by the SP2 scheme, the Fermi function approximation then becomes
    \begin{align}
        \widetilde{F}(\varepsilon) := 1 - \sum_{i=0}^n c_i \varepsilon_i\;.
    \end{align}
    
    Subsequently, the $\theta_{i,j}$ and $c_i$ are trained to minimize the mean squared error over a pre-selected grid, $\{\varepsilon_i\}_{i=0}^N$, on $[0,1]$,
    \begin{equation}
    \mathcal{L}(\widetilde{F}) = \sum_i w_i [\widetilde{F}(\varepsilon_i) - F(\varepsilon_i)]^2\\
    \approx \int_0^1 [\widetilde{F}(\varepsilon) - F(\varepsilon)]^2 \; d\varepsilon \;.
    \end{equation}
    We use the Levenberg-Marquardt (LM) optimization method, which is designed specifically for a sum-of-squares loss function. LM dynamically blends the Gauss-Newton method, yielding quadratic convergence where possible, and gradient descent, slower, but having more robust convergence guarantees. 
    
    Figure \ref{OptTruncSP2} shows an example of a globally optimized truncated SP2 recursive Fermi-operator scheme in comparison to the corresponding Fermi-Dirac function,
    \begin{equation} \label{fermi}
        F(\varepsilon) = \left(e^{\beta(\varepsilon-1/2)} +1 \right)^{-1},
    \end{equation}
    with $\beta = 40$. 
    The approximation error is shown in the lower panel. Previously, we have been limited to the use of recursive Fermi-operator expansions that are based on rational Pade' polynomials as their projections to reach this level of accuracy \cite{ANiklasson03b,ANiklasson15}. However these schemes are implicit and require a solution to a system of equations in each iteration. Here, we are able to achieve a similar level of accuracy using the explicit machine-learned generalized SP2 expansion scheme as presented in Eq.\ (\ref{general SP2}).
    
    \begin{figure}
        \includegraphics[scale=0.31]{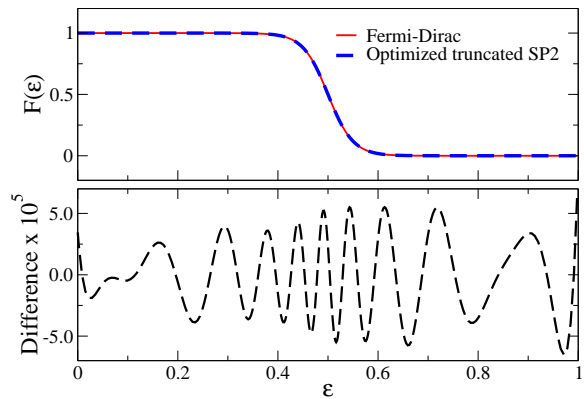}
        \caption{\label{OptTruncSP2}
            {\small Result from learning the Fermi function with 11 Layers and $\beta = 40$. The exact Fermi function (red) is compared with the learned one (dashed blue) in the top panel. The bottom panel shows the error as a function of the scaled energy $\varepsilon$.}}
    \end{figure}

    \section{Conclusions}
    
    We have demonstrated how the solution to the quantum-mechanical electronic structure problem, for example, appearing in Hartree-Fock and Kohn-Sham density functional theory, can be mapped onto the computational structure of a generalized deep neural network. The solution, in terms of an effective single-particle density matrix, is generated by a recursive Fermi-operator expansion derived from a second-order spectral projection scheme. The main computational bottleneck of the layered network is dominated by the activation function, a matrix square operation, which can be performed with high efficiency on tensor core units using a mixed-precision formulation that enhances the intrinsic half-precision floating point operations. A single precision matrix-matrix multiplication in the activation function is replaced by two half-precision matrix-matrix multiplications, allowing us to make full use of available tensor core architectures. This leads to an impressive speed up of about 16x for the calculation of density matrices with respect to the same generation GPUs.
    
    By capitalizing on the machine learning perspective of the deep neural network formulation of the recursive second-order SP2 Fermi-operator expansion, we were able to both accelerate the rate of convergence, by optimizing the weights of the neural net, and apply machine learning techniques to closely approximate Fermi-Dirac functions at finite electronic temperatures.

    \section{Acknowledgments}
    This work is supported by the U.S. Department of Energy Office of Basic Energy Sciences (FWP LANLE8AN), the LANL LDRD-ER program, and by the U.S. Department of Energy through the Los Alamos National Laboratory. We are thankful to Nicolas Bock for his advice on code development.     

    \section{Appendix}
    Here, we state the result used to justify our parameter-free convergence criterion in Alg.\ \ref{Deep_NN_SP2}, which is based on a bound of the worst case error reduction on the general form ${\rm Error}_i \leq C ({\rm Error}_{i-2})^2$ for some constant $C$. We use the estimate of the idempotency error, ${\rm IdErr}_i = \emph{Tr}[S_{i-1} - S_{i-1}^2]$, for the error measure ${\rm Error}_i$. We then determine that convergence occurs once the estimated error reduction (in Eq.\ (\ref{theorem}) below)  no longer holds with the available precision of the floating point operations. It is always valid in exact arithmetics. The theory is analogous to the previous parameter-free convergence criterion by Kruchinina {\em et al.} \cite{AKruchinina16}, which was based on a different measure of the idempotency error.
    \newline
    \begin{theorem}
    Assume that $\sigma_{i} \neq \sigma_{i-1}$ so that either $S_i = (2S_{i-2}-S_{i-2}^2)^2$ or $S_i = 2S_{i-2}^2-S_{i-2}^4$ and $i>1$. Assume also that $S_{i-2}$ has all eigenvalues in $[0, 1]$. Then,
    \begin{align}
    \begin{split} \label{theorem}
       \emph{IdErr}_{i+1} = \emph{Tr}[S_i-S_i^2] & \leq C (\emph{Tr}[S_{i-2}-S_{i-2}^2])^2 \\ 
       &= C (\emph{IdErr}_{i-1})^2 \;,
    \end{split}
    \end{align}
    with $C = \frac{1}{32}(71+17\sqrt{17}) \approx 4.41$.
    \end{theorem}
    
    \begin{proof}
    Let $\{\lambda_j^{(i)}\}_{j=1}^N$ be the eigenvalues of $S_i$, where the ordering of eigenvalues is such that $\lambda_j^{(i)} = \lambda_j^{(i-1)} + \sigma_i(\lambda_j^{(i-1)}-(\lambda_j^{(i-1)})^2), \ j = 1,\dots,N$. From Ref.\ \cite{AKruchinina16} we have that
    \begin{align}
       \max_{\lambda \in (0, 1)}  \frac{(2\lambda-\lambda^2)^2 - (2\lambda-\lambda^2)^4}{(\lambda-\lambda^2)^2} &\\
        = \max_{\lambda \in (0, 1)} \frac{2\lambda^2-\lambda^4 - (2\lambda^2-\lambda^4)^2}{(\lambda-\lambda^2)^2} &= C \;,
    \end{align}
    which, given that $\sigma_i \neq \sigma_{i-1}$, means for $j = 1,\dots,N$
    \begin{align}
        \lambda_j^{(i)}-(\lambda_j^{(i)})^2 \leq C \left( \lambda_j^{(i-2)}-(\lambda_j^{(i-2)})^2\right)^2\;.
    \end{align}
    Summing over all eigenvalues,
    \begin{align}
        \textrm{IdErr}_{i+1}  = &
        \textrm{Tr}[S_i-S_i^2]  \\
        = & \sum_{j=1}^N
        \lambda_j^{(i)}-(\lambda_j^{(i)})^2 \\
        \leq & \sum_{j=1}^N C \bigg( \lambda_j^{(i-2)}-(\lambda_j^{(i-2)})^2\bigg)^2 \\
        = & C \bigg( \bigg( \sum_{j=1}^N \lambda_j^{(i-2)}-(\lambda_j^{(i-2)})^2\bigg)^2 \\ 
        & - \sum_{j=1}^N \sum_{k\neq j} \big(\lambda_j^{(i-2)}-(\lambda_j^{(i-2)})^2\big) \nonumber \\
        & \qquad \times \big(\lambda_k^{(i-2)}-(\lambda_k^{(i-2)})^2 \big)\bigg) \nonumber \\
        \leq & C \left(\textrm{Tr}[S_{i-2}-S_{i-2}^2]\right)^2 \\
        = & C (\textrm{IdErr}_{i-1})^2 
        \;.
    \end{align}
    \end{proof}
    Note that $C$ is not the asymptotic error constant, but since $C$ is finite, the result implies the established quadratic convergence for sequences with alternating polynomials in the limit of idempotent matrices $S_i$ \cite{ANiklasson02,EHRubensson14}.

\bibliography{manuscript.bib}
\end{document}